\documentclass[11pt]{article}
\usepackage[utf8]{inputenc}
\usepackage{float}
\usepackage{listings}
\usepackage{algorithm}
\usepackage[noend]{algorithmic}
\usepackage{amsmath, amsthm, amssymb, amsfonts}
\usepackage{graphicx}
\usepackage{xy}
\usepackage{fullpage}
\usepackage{indentfirst}
\usepackage{color}

\newtheorem{theorem}{Theorem}[section]

\newtheorem{corollary}[theorem]{Corollary}

\newtheorem{axiom}[theorem]{Axiom}

\newcommand{\set}[1]{\{#1\}}

\newcommand{\R}{\mathbb{R}}

\newcommand{\lastproc}{n}
\newcommand{\lastitem}{380}

\newcommand{\procs}{1, \ldots, \lastproc}

\newcommand{\setprocs}{\set{\procs}}

\newcommand{\cI}{{\mathcal{I}}}

\newcommand{\cP}{{\mathcal{P}}}
\newcommand{\procinput}[1]{\mathbb{I}({#1})}
\newcommand{\protinput}[1]{\mathbb{I}({#1})}

\newcommand{\algorithminput}{\mathbb{I}}
\newcommand{\universeinputs}{\mathbb{I}^*}

\newcommand{\algorithmoutput}{\mathbb{O}}

\newcommand{\sameorder}[4]{\mathbf{SameOrder}({#1},{#2},{#3} \sim {#4})}
\newcommand{\samevalue}[4]{\mathbf{SameValue}({#1},{#2},{#3} \sim {#4})}

\newcommand{\orderHigher}[3]{{#2} \succ_{#1} {#3}}
\newcommand{\orderLower}[3]{{#2} \prec_{#1} {#3}}
\newcommand{\orderEqual}[3]{{#2} =_{#1} {#3}}
\newcommand{\orderHigherEqual}[3]{{#2} \succeq_{#1} {#3}}
\newcommand{\orderLowerEqual}[3]{{#2} \preceq_{#1} {#3}}

\newcommand{\inst}[2]{{\mathit{Inst}_{{#1}}({#2})}}

\graphicspath{{./figures/}}

\title{Thermodynamic Hypothesis as Social Choice: An Impossibility Theorem for Protein Folding}
\author{Hammurabi Mendes, Sorin Istrail\\
\texttt{\{hmendes,sorin\}}@cs.brown.edu\\
Brown University}
\date{}

\begin{document}

\maketitle

\begin{abstract}
Protein Folding is concerned with the reasons and mechanism behind a protein's tertiary structure. The thermodynamic hypothesis of Anfinsen postulates an universal energy function (UEF) characterizing the tertiary structure, defined consistently across proteins, in terms of their aminoacid sequence.

We consider the approach of examining multiple protein structure descriptors in the PDB (Protein Data Bank), and infer individual preferences, biases favoring particular classes of aminoacid interactions in each of them, later aggregating these individual preferences into a global preference. This 2-step process would ideally expose intrinsic biases on classes of aminoacid interactions in the UEF itself. The intuition is that any intrinsic biases in the UEF are expressed within each protein in a specific manner consistent with its specific aminoacid sequence, size, and fold (consistently with Anfinsen's thermodynamic hypothesis), making a 1-step, holistic aggregation less desirable.

Our intention is to illustrate how some impossibility results from voting theory would apply in this setting, being possibly applicable to other protein folding problems as well. We consider concepts and results from voting theory and unveil methodological difficulties for the approach mentioned above. With our observations, we intend to highlight how key theoretical barriers, already exposed by economists, can be relevant for the development of new methods, new algorithms, for problems related to protein folding.
\end{abstract}

\newpage

\section{Introduction}
\label{Sec-Introduction}

The systematic aggregation of individual preferences into a global social preference is the principal concern of social choice theory. Its origins trace back to the 13th-century, when the philosopher Ramon Lull considered a voting method that counted pairwise preferences among candidates. This technique was later recovered by the famous 18th-century mathematicians Borda and Condorcet. Later, in 1951, Kenneth Arrow stated his famous impossibility theorem (\cite{arrow}, for the version of 1963), which is viewed as a landmark of social choice theory. This important result defines apparently innocuous properties that certain preference aggregation functions (called social welfare functions) should respect, but also demonstrates that such properties are conflicting.

Many other impossibility results followed. The conclusions usually result from an axiomatic setting, where aggregation functions are modeled after axioms defining their intended behavior. Also common is the relative unconcern for the elements subject to preferences. Many aggregation settings, in addition, do \textbf{not} consider voters as players, but merely as providers of information -- and we are interested on those settings \textbf{only}. Note that the observations above suggest a potential application of the impossibility results to other circumstances that involve aggregation. In this work, we model a problem related to protein folding as an aggregation function prone to those impossibility results, exploring the connection between these ideas.

The protein folding problem is basically interested on how and why proteins settle in particular three-dimensional structures. The well-accepted \emph{thermodynamic hypothesis} of Anfinsen \cite{Anfinsen} states that the protein structure, informally speaking, is a unique, stable, and viable state depending solely on \emph{aminoacid sequence}. The state minimizes an energy function accounting for physiochemical utilities. Such \emph{universal energy function} (UEF) is defined consistently across proteins, in terms of the particular aminoacid interactions that define each of them.


The protein folding literature often takes a \emph{holistic} approach to analyzing the PDB~\cite{PDB} in order to obtain certain conclusions from stable folds. We define a problem related to protein folding where a 2-step aggregation process will seem more intuitive: suppose we intend to infer hypothesized intrinsic preferences on the UEF, which would favor certain classes of aminoacid interactions over others. The idea is to first aggregate multiple \emph{individual preferences}, and later consolidate them into a single \emph{global preference}. In other words, we consider and explore the following scenario:
\begin{itemize}
	\item Say that the UEF has an intrinsic \emph{global preference} among interaction classes. For instance, the UEF might generally favor $(A,V)$, the interaction of the aminoacids alanine and valine, over $(G,L)$, the interaction of aminoacids glycine and leucine.
	\item This global preference is manifested distinctly across proteins, in light of the thermodynamic hypothesis. In other words, a protein's stable fold implies an \emph{individual preference} among interaction classes, which results from its corresponding aminoacid sequence and the global preference intrinsic to the UEF.
\end{itemize}

We consider the possibility of inferring or estimating the postulated \emph{global preference} among interaction classes, intrinsic to the UEF, by aggregating \emph{individual preferences}, extracted separately from multiple stable folds in the PDB. We assess and qualify aminoacid interactions inside each protein in order to generate the individual preferences. Importantly, we assume that aminoacid interactions are totally defined by \emph{pairwise interactions} (see Sec.~\ref{Sec-PairwiseInteractions}). The existence and expression of a pairwise interaction depends on the aminoacids involved (their physiochemical affinities) and their spacial arrangement (their distance and contact).

We particularly think that such 2-step aggregation process is particularly sensible in light of Anfinsen's thermodynamic hypothesis. The idea is that individual proteins may express universal biases differently, according to their particular aminoacid sequence, size, and fold (note here the thermodynamic hypothesis). If we holistically analyze the PDB, performing instead a 1-step aggregation process, we effectively aggregate over aminoacid interactions from multiple proteins. By assumption, those aminoacid interactions may express biases differently. This argument might as well be applicable to other problems related to protein folding where the thermodynamic hypothesis is considered.

The problem formalized before is formally \emph{weaker} than those aiming for more complete characterizations of the UEF with our 2-step aggregation process and the same axiomatic setting. Here, while we do not intend to completely characterize the UEF, we do intend to identify intrinsic biases among interaction classes in the UEF. Unfortunately, this weaker problem will still manifest an impossibility in light of some seemingly reasonable assumptions (discussed in Sec.~\ref{Sec-AnalyticalFramework} and Sec.~\ref{Sec-ExternalAggregation}). Stronger characterizations of the UEF via the method above will suffer from similar methodological limitations unless axioms are relaxed, made more specific, or otherwise changed appropriately.

Our intention and contribution is to highlight theoretical impossibility results for the problem considered above, under reasonable assumptions (see Sec.~\ref{Sec-ExternalAggregation}). We use concepts and results from voting theory, a distinct investigation field of social choice theory, in Economics. The goal is to sparkle interest in the profound results from voting theory, and highlight a connection of this area with methods that, so far, have been mainly a concern of biologists, chemists, physicists, and computer scientists.
We note that the remarks made here result from a suitable axiomatic framework, from the application of voting theory concepts, and from direct reduction to previous results in this area. We intend to suggest how some key theoretical barriers, already exposed by economists, can be relevant for the development of new methods, new algorithms, for problems related to protein folding, particularly those in which the thermodynamic hypothesis is considered.

\section{The Analytic Framework}
\label{Sec-AnalyticalFramework}

In this section, we present the axioms and definitions that form our analytic framework. We adhere to the thermodynamic principle of \cite{Anfinsen}, and further assume that aminoacid interactions in the universal energy function are totally defined by pairwise interactions. We try to abstract concepts of purely biological nature as much as possible. Again, the purpose is to bring impossibility results from voting theory to our 2-step preference model and suggest that they might be relevant to other problems related to protein folding.

\subsection{Thermodynamic Hypothesis}
\label{Sec-ThermodynamicHypothesis}

There are two principal questions in the Protein Folding research \cite{ProtFolding}. The first one concerns the reasons why proteins settle in a specific tertiary structure. The second asks how to correctly predict the tertiary structure given the protein's aminoacid configuration.

The thermodynamic hypothesis \cite{Anfinsen} says that the tertiary structure results solely from \emph{aminoacid interactions}, maximizing a function of utility character and physiochemical nature. Also, the final structure is (1) \emph{unique}, as other possible states have strictly lower energy; (2) \emph{stable}, as small external changes do not disrupt the structure; and (3) \emph{viable}, smoothly accessible by physiochemical changes.

\begin{axiom}
The thermodynamic hypothesis of \cite{Anfinsen} is valid, hence the protein's tertiary structure is unique, stable, viable, and depends solely on aminoacid interactions.
\end{axiom}

\subsection{Pairwise Interactions}
\label{Sec-PairwiseInteractions}

Previous literature suggests energy functions depending on number, distance, and additional physiochemical properties (hydrophobia/hydrophilia, among others) between \emph{pairs} of aminoacids in the three-dimensional protein structure \cite{energy1, energy2, estimation1}. We will follow such approach, hence \emph{pairwise interactions} fully characterize our aminoacid interactions in this work.

Identifying pairwise interactions instead of trying to obtain all collective, grouped interactions is much more viable computationally. Besides, it is conceivable that more intricate interaction patterns be extrapolated from pairwise interactions. In this work, we consider pairwise energy functions, i.e., those totally defined by energy potentials on pairwise interactions among aminoacids.

\begin{axiom}[Pairwise Hypothesis]
\emph{All} existent information on aminoacid interactions \emph{can} be obtained by considering \emph{only} pairwise interactions.
\end{axiom}

\subsubsection{Interaction Classes and Instances}
\label{Sec-InteractionClassesInstances}

An \emph{interaction class} is a pair $(X_1, X_2)$ where $X_1$ and $X_2$ are \emph{names} of aminoacids. For instance, $(A,V)$ is the interaction of the aminoacids alanine and valine, and $(G,L)$ is the interaction of aminoacids glycine and leucine. We consider $(X_1, X_2)$ equal to $(X_2, X_1)$ for any aminoacid names $X_1,X_2$. Considering the standard 20 aminoacids that characterize proteins, we have 380 possible interactions classes.

Within a protein's stable fold, an \emph{interaction instance} of $(X_1,X_2)$ \emph{occurs} if the distance between two aminoacids of types $X_1$ and $X_2$, is within an arbitrary threshold $\tau$. In practice, the distance could be taken over heavy atoms, centroids in aminoacids, or $c_\alpha$ locations in aminoacids (the location of the central carbon atom). Our results are independent of the choice above. We represent interaction classes with uppercase letters (e.g. $I$), and interaction instances with lowercase letters (e.g. $i$). In a protein's stable fold $P$, the interaction instances of the interaction class $I$ defines the set $\inst{P}{I}$.

Establishing a distance threshold might seem arbitrary and even unrealistic, as overlooking interaction instances that barely fail such distance may discard useful information. Also, whatever account attributed to interaction instances must also consider certain physiochemical properties, for instance, hydrophobia or hydrophilia\footnote{Respectively, the tendency or aversion to bind to water.}, usually compiled into tables such as the Miyazawa-Jernigan~\cite{estimation1}. One could argue that the right choice of distance threshold and tables describing physiochemical interactivity could be enough to totally define the UEF. Let us fix (yet not specify) $\tau$ big enough as to accommodate all possible \emph{actual} aminoacid interactions, and \emph{assume} that contacts are scored accordingly.

Therefore, to each interaction instance $i$, we associate an \emph{energetic score} $e(i) \in \R$, which abstracts the following parameters: (1) distance threshold $\tau$; (2) contact point; (3) the physiochemical properties of the aminoacids involved. As far as identifying intrinsic preferences over interaction classes in the UEF, our results are independent of the above choices -- in other words, independent of the energetic score function. The key to such generality of results lies in our 2-step preference model, in our axiomatic setting, and impossibility results from voting theory.

\section{Internal Aggregation}
\label{Sec-InternalAggregation}

The \emph{internal aggregation} is the procedure where the energetic scores of the interaction instances are taken into account and define a protein's individual preference. For any protein $P$:
\begin{itemize}
	\item Each interaction instance $i$ has an energetic score $e(i) \in \R$ modeling distance threshold, contact point, and physiochemical properties of $i$. The function $e()$ is fixed yet arbitrary. The value of $e(i)$ can be positive or negative, modeling physiochemical properties that might manifest with opposing effect (for instance, hydrophobia vs. hydrophilia).
	\item Each interaction class $I$ has a set of interaction instances $\inst{P}{I}$. The individual preference of $P$ regarding $I_1$ versus $I_2$ depends \emph{only} on $\inst{P}{I_1}$ and $\inst{P}{I_2}$, and the energetic scores respectively associated with each interaction instance. We consider two scenarios:
	\begin{enumerate}
		\item For each interaction instance $I$, we can infer an utility value $u_P(I) \in \R$, completely defined in terms of $\inst{P}{I}$. Considering the same protein/function, the utilities are comparable: if $u_P(I_1) > u_P(I_2)$, then $P$ manifests preference (i.e., bias) for $I_1$ instead of $I_2$; if $u_P(I_1) = u_P(I_2)$, then $P$ does not manifest any preference (i.e., bias) for $I_1$ instead of $I_2$.
	
		
Note that we do \emph{not} assume comparability between utility levels across two different proteins. This is what characterizes our 2-step preference model: utilities are meaningful to its corresponding protein only. In Sec.~\ref{Sec-UtilityPreferenceModel} we in fact vary the comparability between utility levels across proteins, and explore the consequences of each mode.

		\item We can only infer a total order $\preceq_P$ among interaction classes, completely defined in terms of $\inst{P}{I_1}$ and $\inst{P}{I_2}$. If $I_1 \succ_P I_2$, then $P$ prefers $I_1$ to $I_2$; if $I_1 \prec_P I_2$, then $P$ prefers $I_2$ to $I_1$; if $I_1 =_P I_2$, then $P$ is indifferent regarding $I_1$ versus $I_2$.
	\end{enumerate}
\end{itemize}


We call the first scenario the \emph{utility} preference model, and the second scenario is an \emph{ordinal} preference model (our terminology). We examine the two scenarios separately, in each case we apply the appropriate concepts and impossibility results from social choice theory. Recapitulating, our goal is to evaluate the viability of the following strategy:
\begin{description}
	\item[Internal Aggregation:] Examine interaction instances of each protein, deriving the manifested \emph{individual preferences}, either a utility function $u_P$ or a total order $\preceq_P$, for each protein $P$.
	\item[External Aggregation:] Examine individual preferences and derive the \emph{global preference}, which is similarly a utility function or a total order.
\end{description}

We assume that the internal aggregation might be subject to errors, but not in presence or scope such that they might deem any individual preference ``more representative'' than another. Errors are small and covert. This idea is formalized by the anonymity requirement (Axiom~\ref{axiom-anonimity}) for the external aggregation.

\section{External Aggregation}
\label{Sec-ExternalAggregation}

In this section we will axiomatize the requirements of the external aggregation procedure. We discuss the general problem structure on Sec.~\ref{Sec-BasicNotation}, followed by an axiomatic formulation of the basic requirements on Sec.~\ref{Sec-GeneralAxioms}. We finally show that the external aggregation (axiomatized as below) is subject to fundamental impossibility results studied in the context of voting theory regardless of: (1) the function $e$; and (2) the utility function $u_P$ or the total order $\preceq_P$ for any protein $P$.

In other words, the strategy outlined above, in the end of Sec.~\ref{Sec-InternalAggregation}, will pose some fundamental methodological limitations \emph{regardless} on how the internal aggregation is defined -- how distances, contact, physiochemical potentials, as well how preferences are extracted. Specific axioms and impossibility theorems for the case with ordinal preference is presented on Sec.~\ref{Sec-OrdinalPreferenceModel}, and for the case with utility preference on Sec.~\ref{Sec-UtilityPreferenceModel}.

\subsection{Basic Notation}
\label{Sec-BasicNotation}

The external aggregation is defined as follows. Consider all interaction classes $\cI = \set{I_1, \ldots, I_{\lastitem}}$, and arbitrary $n$ proteins $\cP = \set{P_1, \ldots, P_n}$, with some finite $n \ge 2$. Considering our two preference models:
\begin{itemize}
	\item In the ordinal preference setting, the individual preference of $P_i$ is a \emph{total order} $\preceq_i$ over $\cI$. By definition, these relations are total ($I_a \preceq_i I_b$ or $I_b \preceq_i I_a$ for any $I_a,I_b$), transitive (if $I_a \preceq_i I_b$ and $I_b \preceq_i I_c$ then $I_a \preceq_i I_c$), and antisymmetric ($I_a \preceq_i I_b$ and $I_b \preceq_i I_a$ implies $I_a =_i I_b$).

	\item In the utility preference setting, the individual preference of $P_i$ is a \emph{utility function} $u_i: \cI \to \R$. The utilities imply a total order (if $u_i(I_a) \le u_i(I_b)$ then $I_a \preceq_i I_b$).
\end{itemize}

The \emph{individual preference} manifested by protein $P_i$ is denoted as $\procinput{i}$. The \emph{preference profile} is composed of all individual preferences: $\algorithminput = (\procinput{0}, \ldots, \procinput{\lastproc})$. Finally, the \emph{universe of preference profiles} is composed of all preference profiles: $\universeinputs = \set{\algorithminput: \text{$\algorithminput$ is a valid preference profile}}$. The external aggregation procedure returns a \emph{global preference}, either a total order $\preceq_\algorithmoutput$ or a utility function $u_\algorithmoutput$ over $\cI$. We consider the two cases separately.

In terms of notation, for any individual or the global preference $X$, if $I_a \preceq I_b$ in $X$, we write $\orderLowerEqual{X}{I_a}{I_b}$ (and similarly for other relations). So, if $I_a \preceq I_b$ in $\procinput{i}$, we write $\orderLowerEqual{\procinput{i}}{I_a}{I_b}$; besides, if $I_a \preceq I_b$ in $\algorithmoutput$, we write $\orderLowerEqual{\algorithmoutput}{I_a}{I_b}$. For any two preferences $X$ and $Y$, if $\orderHigherEqual{X}{I_a}{I_b} \Leftrightarrow \orderHigherEqual{Y}{I_a}{I_b}$, then we write $\sameorder{X}{Y}{I_a}{I_b}$.

\subsection{General Axioms}
\label{Sec-GeneralAxioms}

Our first axioms simply models the requirement that the procedure receives multiple individual preferences and produces a single global preference. Individual preferences are transitive ordering over interaction classes.
\begin{axiom}
\label{axiom-agree}
\textbf{Agreement} After the external aggregation, we obtain a single global preference among interaction classes.
\end{axiom}

\begin{axiom}
\label{axiom-transitivity}
\textbf{Transitivity} The individual preferences as well as the global preference are total transitive orderings over interaction classes.
\end{axiom}

The unrestricted domain requirement, stated below, is an imposition on the external aggregation, not on the actual preference profile. Intuitively speaking, the external aggregation must be able to handle arbitrary transitive individual preferences.
\begin{axiom}
\label{axiom-unrestdomain}
\textbf{Unrestricted Domain} Any total order among interaction classes is permitted in a preference profile.
\end{axiom}

If any set of at least three interaction classes $S = \set{I_1,I_2,I_3}$ is such that both (1) $I_i \succeq I_j$ and (2) $I_j \succeq I_i$ are possible preferences in the preference profile, for any $1 \le i,j \le 3$, $S$ is called a 3-unrestricted preference set.

In a first glance, Axiom~\ref{axiom-unrestdomain} may appear overly imposing for practical purposes. After all, the preference profile expressed in the PDB, given a suitable internal aggregation procedure, can be more restrictive. However, it seems reasonable to expect a 3-unrestricted preference set in the preference profile after the internal aggregation. In that case, the Arrow's Impossibility Theorem (Theorem~\ref{theorem-arrow}) is applicable, and May's Theorem (Theorem~\ref{theorem-may}) may imply the possibility of breaking transitivity (discussed ahead). Even considering the much less restrictive condition of having only a 3-unrestricted preference set in the preference profile, the results apply. We keep the complete unrestricted domain axiom for its generality appeal.

The unanimity requirement requires that if all proteins expose a preference on a particular interaction class over another, the unanimity should be respected.
\begin{axiom}
\label{axiom-unanimity}
\textbf{Unanimity.} For any $\algorithminput \in \universeinputs$, if $\sameorder{\procinput{i}}{\procinput{j}}{I_a}{I_b}$ for all $P_i,P_j$ then $\sameorder{\procinput{i}}{\algorithmoutput}{I_a}{I_b}$ for all $P_i$, considering arbitrary $I_a,I_b$;
\end{axiom}

We assume that proteins manifest preferences over interaction classes in an indistinguishable manner; i.e., we cannot infer whether protein $A$ manifests a ``better'' individual preference than protein $B$. As mentioned before, we assume that the internal aggregation might be subject to errors, but not in presence or scope such that they might deem any individual preference ``more representative'' than another. Errors are small and covert. In specific terms, the preferences manifested by protein $A$ and those by protein $B$ are considered without regard to their identity.
\begin{axiom}
\label{axiom-anonimity}
\textbf{Anonymity.} For any $\algorithminput',\algorithminput'' \in \universeinputs$, whenever $\algorithminput' = (\procinput{0}', \ldots, \procinput{\lastproc}')$ is a permutation of $\algorithminput'' = (\procinput{0}'', \ldots, \procinput{\lastproc}'')$, we have that $\algorithmoutput' = \algorithmoutput''$.
\end{axiom}

We note that anonymity implicates in non-dictatorship, i.e., the fact that no individual preference completely and exclusively defines the global preference.
\begin{axiom}
\label{axiom-ndict}
\textbf{Non-Dictatorship.} For all $P_i \in \cP$ there exists at least one $\algorithminput \in \universeinputs$ and elements $I_a, I_b$ such that $\orderHigher{\procinput{i}}{I_a}{I_b}$ but $\orderLowerEqual{\algorithmoutput}{I_a}{I_b}$.
\end{axiom}

\subsection{Ordinal Preference Model}
\label{Sec-OrdinalPreferenceModel}


Our external aggregation procedure is modeled after \emph{social welfare functions}, a basic concept in social choice theory, simply representing a procedure that receives total orders on elements of $\cI$ and produces a total order over elements of $\cI$. We start with basic axioms meaningful for the ordinal preference model, followed by impossibility results.

\subsubsection{Axioms}

The neutrality requirement, formalized below, basically requires that all pairwise preferences over interaction classes are treated equally: the procedure computing the outcome of all pairwise preferences is oblivious to the particular names of interaction classes. We say that two preference profiles are \emph{compatible} whenever they have the same number of individual preferences.

\begin{axiom}
\label{axiom-neutrality}
\textbf{Neutrality.} For any compatible $(\algorithminput',\algorithminput'')$ and corresponding outputs $(\algorithmoutput',\algorithmoutput'')$, if we have $[ \forall P_i \, e_a \succeq_{\procinput{i}'} e_b \Leftrightarrow e_c \succeq_{\procinput{i}''} e_d ]$ then $e_a \succeq_{O'} e_b \Leftrightarrow e_c \succeq_{O''} e_d$, considering arbitrary $e_a,e_b,e_c,e_d$;
\end{axiom}

If we make $e_c = e_a$ and $e_d = e_b$, neutrality implies the famous and controversial Independence of Irrelevant Alternatives (IIA). This requirement basically says that for any two aggregation scenarios where individual preferences over arbitrary elements $e_a$ and $e_b$ remain identical, the algorithm output similarly remains identical over those elements.

\begin{axiom}
\label{axiom-iia}
\textbf{Independence of Irrelevant Alternatives (IIA).} For any compatible $\algorithminput',\algorithminput'' \in \universeinputs$ and corresponding outputs $\algorithmoutput',\algorithmoutput''$, if we have $[ \forall P_i \, \sameorder{\procinput{i}'}{\procinput{i}''}{e_a}{e_b} ]$ then we have $\sameorder{\algorithmoutput'}{\algorithmoutput''}{e_a}{e_b}$, considering arbitrary $e_a,e_b$;
\end{axiom}

Intuitively speaking, we want an aggregation procedure that behaves similarly to continuous functions in the following sense: small changes in the preference profile will not induce substantial changes in the global preference. We formalize this requirement using distance metrics. We first define a distance metric $d: \universeinputs \times \universeinputs \rightarrow \R$, which naturally satisfies the following:
\begin{enumerate}
	\item $d(X,X) = 0$
	\item $d(X,X') = d(X',X)$
	\item $d(X,X'') \le d(X,X') + d(X',X'')$
\end{enumerate}
for any $X,X' \in \universeinputs$. Then, between two compatible preference profiles $\algorithminput', \algorithminput''$, we define $$ D = \sum_{i \in \setprocs} d(\procinput{i}, \procinput{i}') \mathrm{.} $$
Technically, we want smaller changes in the preference profile to produce smaller changes in the global preference, where changes are measured according to the distance function. The motivation is having an external aggregation that is \emph{robust} to small imprecision in the internal aggregation. Formally:
\begin{axiom}
\label{axiom-proximity}
\textbf{Proximity Preservation.} There exists a distance function $d$ over $\universeinputs$ such that $$ D(\algorithminput, \algorithminput') \le D(\algorithminput, \algorithminput'') \Rightarrow d(O, O') \le d(O, O'') \mathrm{.} $$
\end{axiom}

We want an external aggregation that responds positively (or at least monotonically) to the preferences of the proteins. We want that positive changes on the preference of $I_1$ over $I_2$, say, due to improving our internal aggregation, induce positive (or at least monotonic) changes in the global order. The two axioms are presented below.

\begin{axiom}
\label{axiom-posrep}
\textbf{Positive Responsiveness.} For any $\algorithminput',\algorithminput'' \in \universeinputs$ and corresponding outputs $\algorithmoutput',\algorithmoutput''$, if $\sameorder{\procinput{i}'}{\procinput{i}''}{e_a}{e_b}$ for all $i \ne j$, and, for some $j$
\begin{enumerate}
	\item $\orderLower{\procinput{j}'}{e_a}{e_b} \Rightarrow \orderHigherEqual{\procinput{j}''}{e_a}{e_b}$;
	\item $\orderEqual{\procinput{j}'}{e_a}{e_b} \Rightarrow \orderHigher{\procinput{j}''}{e_a}{e_b}$,
\end{enumerate}
then $\orderHigherEqual{\algorithmoutput'}{e_a}{e_b} \Rightarrow \orderHigher{\algorithmoutput''}{e_a}{e_b}$, considering arbitrary $e_a,e_b$.
\end{axiom}

\begin{axiom}
\label{axiom-monrep}
\textbf{Monotonic Responsiveness.} For any $\algorithminput',\algorithminput'' \in \universeinputs$ and corresponding outputs $\algorithmoutput',\algorithmoutput''$, if $\sameorder{\procinput{i}'}{\procinput{i}''}{e_a}{e_b}$ for all $i \ne j$, and, for some $j$
\begin{enumerate}
	\item $\orderLower{\procinput{j}'}{e_a}{e_b} \Rightarrow \orderHigherEqual{\procinput{j}''}{e_a}{e_b}$;
	\item $\orderEqual{\procinput{j}'}{e_a}{e_b} \Rightarrow \orderHigher{\procinput{j}''}{e_a}{e_b}$,
\end{enumerate}
then $\orderHigherEqual{\algorithmoutput'}{e_a}{e_b} \Rightarrow \orderHigherEqual{\algorithmoutput''}{e_a}{e_b}$, considering arbitrary $e_a,e_b$.
\end{axiom}

\subsubsection{Impossibility Results}

In this section, we romp through impossibility results from social choice theory and expose contradictions hidden in our previous axioms. First, considering the results from \cite{ProximityPreservation}, we conclude that any external aggregation procedure that respects unanimity and satisfies anonymity does not preserve proximity.
\begin{theorem}
\label{theorem-proximity}
No external aggregation can simultaneously respect Axioms~\ref{axiom-unanimity}, \ref{axiom-anonimity}, \ref{axiom-proximity}.
\end{theorem}
\begin{proof}
Proof in \cite{ProximityPreservation}.
\end{proof}

Theorem~\ref{theorem-proximity} exposes a methodological difficulty: an external aggregation procedure that respects anonymity and unanimity does not preserve proximity, which goes against the intuition that smaller changes in the preference profile should imply in smaller changes in the global preference. Note that the theorem in~\cite{ProximityPreservation} is valid even when the unrestricted domain axiom is not considered.
However, if we waive proximity preservation (hence keeping unanimity and anonymity), having an unrestricted domain still frustrates our requirements, as we conflict with neutrality:
\begin{corollary}
\label{corollary-arrow}
No external aggregation can simultaneously respect Axioms~\ref{axiom-transitivity}, \ref{axiom-unrestdomain}, \ref{axiom-unanimity}, \ref{axiom-anonimity}, \ref{axiom-neutrality}.
\end{corollary}
\begin{proof}
Since anonymity implicates in non-dictatorship and neutrality implicates in the IIA, this is a corollary of the Arrow's Impossibility Theorem~\cite{arrow}.
\end{proof}

The famous Arrow's Impossibility Theorem \cite{arrow} is widely considered as a landmark of modern social choice theory. There exist many proofs for such result (including \cite{arrow,arrovian-imposresults,arrow-three-proofs}). The original theorem is presented below. Note that even relaxing anonymity to non-dictatorship, and neutrality to the IIA, the original theorem applies:

\begin{theorem}[Arrow's Impossibility Theorem]
\label{theorem-arrow}
No external aggregation (modeled after social welfare functions) can simultaneously respect Axioms~\ref{axiom-transitivity}, \ref{axiom-unrestdomain}, \ref{axiom-unanimity}, \ref{axiom-ndict}, \ref{axiom-iia}.
\end{theorem}
\begin{proof}
Proof in \cite{arrow}.
\end{proof}



If we want a positively responsive aggregation procedure (i.e., respecting axiom~\ref{axiom-posrep}), May's Theorem \cite{MayTheorem} says that the preference regarding \emph{each} pairwise comparison among interaction classes is necessarily decided by simple majority. For any interaction classes $I_a,I_b \in \cI$, define
$$ N(I_a > I_b) = |\set{P_i: I_a \succ_{\protinput{i}} I_b}| \mathrm{.} $$

\begin{theorem}[May's Theorem]
\label{theorem-may}
For each interaction class $(I_a,I_b)$, any decision function\footnote{For us, external aggregation procedure.} respecting Axioms~\ref{axiom-unanimity}, \ref{axiom-neutrality}, and~\ref{axiom-posrep} is such that $I_a \succeq_{\algorithmoutput} I_b \Leftrightarrow [N(I_a > I_b) \ge N(I_b > I_a)]$.
\end{theorem}
The problem implied by May's Theorem is that the global preference might be \emph{intransitive} \cite{arrovian-imposresults,social-choice-survey} if a 3-unrestricted preference set exists in the preference profile (which seems reasonable after the internal aggregation). In fact, even waiving neutrality (preserving just the IIA) conflicts with positive responsiveness (\cite{social-choice-survey}, chap. 22), so perhaps the problem is the positive responsiveness itself. Settling for monotonicity (Axiom~\ref{axiom-monrep}) works around May's theorem, but then we could argue whether waiving neutrality (preserving just the IIA) is reasonable or not.


At first, it sounds reasonable to treat preferences consistently across different interaction classes (as suggested by neutrality and by the IIA). As implied above, restrictions like neutrality (and its weaker restriction, the IIA), even with seemingly natural meaning, have strong axiomatic implications. Nevertheless, dropping neutrality poses serious methodological difficulties for our strategy: in that case, \emph{names} of interaction classes matter for our external aggregation. Their relationship and place, however, are precisely our unknowns.

Many impossibility theorems in the voting theory literature presume unrestricted domains (or simply a 3-unrestricted preference set). For the Arrow's Theorem as well as the May's Theorem, performing suitable \emph{domain restrictions}, where we limit the preference profile, might be enough to overcome the theoretical impossibilities. If preferences are \emph{single-peaked} or \emph{quasi-transitive} (please refer to~\cite{Sen66,Sen69}), we also dodge from the impossibility results above. We invite the reader to appreciate these previous results and understand the relevance of these theorems for problems related to protein folding, and perhaps in other similar contexts where a 2-step preference model is similarly suitable.

\subsection{Utility Preference Model}
\label{Sec-UtilityPreferenceModel}

In this section, we discuss the utility preference model. In this setting, it is possible to obtain utility values associated with each interaction class in each individual protein. In other words, for any protein $P_i$, there exists a utility function $u_i$ mapping interaction classes to utility values. Recall that we intend to analyze preferences manifested by individual proteins in an individual setting, and \emph{then} make an external aggregation. We presume utility functions that are not comparable across proteins. In the following subsection, we improve the comparability and explore the consequences.

For now, we technically assume that the utility function for $P$ differ in origin and scale from the utility function for $Q$, taking any proteins $P \ne Q$. Formally:
\begin{axiom}
\label{axiom-scaleorigin}
Any utility function is a strictly increasing linear transformation of another utility function: $u_P = \alpha_{pq} u_Q + \beta_{pq}$ for some $\alpha_{pq} > 0$ and any $\beta_{pq}$, considering any proteins $P$ and $Q$.
\end{axiom}
Borrowing terminology from voting theory, there is no interpersonal comparisons of utility levels or of gains and losses among different individual preferences. Note that this is fundamentally different from the approach where there is free comparison among utility levels (energy scores) measured across different proteins, as seen in~\cite{energy1,energy2,estimation1}.

With that in mind, we are still going to consider ordinal preferences, however expressed through utility functions, in a continuous manner. Consider the space $\R^{380}_{+}$ where the $i$-th coordinate of an arbitrary point $p$ is represented by $p[i]$. A point $p$ an \emph{expression descriptor} for interaction classes: $p$ is such that the bigger the $p[i]$, the bigger the associated expression for the interaction class $I_i$.

Furthermore, for sufficiently close levels of expression descriptors (i.e., points in $\R^{380}_{+}$), the utility function has sufficiently close values. Formally, we define $$ U_P(p) = \sum_{I_1 \ldots I_{\lastitem}} u_P(I_i) \cdot p[i] \mathrm{,} $$ representing the utility associated with some point $p$. Technically, $U_P$ is a utility function over the expression descriptor space. Note that $U_P$ is linear, implying that the expression of $I_1$ is worth $u_P(I_1)/u_P(I_2)$ times the expression of $I_2$ for any interaction classes $I_1$ and $I_2$ and protein $P$.

Of course, this informational framework is subject to the same impossibility results as before, but other interesting results further apply. We now briefly overview the approach of~\cite{Chichilnisky1982,Baigent2011}, describing an impossibility result that mirrors (in the continuous setting) Theorem~\ref{theorem-proximity}.

The indifference hyperplane $H_P^c$ is the hyperplane containing points $p$ where $U_P(p) = c$. Note that $H_P^c$ and $H_P^d$ are parallel for any positive constants $c$ and $d$. The unit vector perpendicular to $H_P^c$ (for an arbitrary constant $c > 0$), called $v_P$, represents the \emph{direction} where the utility \emph{increases} for $P$. The length of the vector is disconsidered precisely because we are taking a purely ordinal interpretation of the utility functions: their origin and scale are not comparable among different proteins.

The individual preference for $P$, denoted by $v_P$, is a unit vector indicating the direction that maximizes $U_P$. Note that $v_P$ denotes a unique \emph{point} in the $380$-dimensional sphere $S^{380}$. The global preference, denoted by $v_{\algorithmoutput}$, and is also a unit vector indicating the direction that maximizes the global utility. The external aggregation is such that
\begin{displaymath}
E: \underbrace{S^{380} \times \ldots \times S^{380}}_{n} \to S^{380} \mathrm{,}
\end{displaymath}
where $E$ satisfies the following:
\begin{axiom}
\label{axiom-continuity}
\textbf{Continuity.} The external aggregation $E$ is a continuous function.
\end{axiom}

The anonymity and unanimity, central requirements for our aggregate mechanism, are rephrased for our continuous setting below.
\begin{axiom}
\label{axiom-unanimity-cont}
\textbf{Continuous unanimity.} If $v_P = v$ for all proteins $P$, then the $v_\algorithmoutput = v$.
\end{axiom}

\begin{axiom}
\label{axiom-anonimity-cont}
\textbf{Continuous anonymity.} $E(v_1, \ldots, v_{\lastproc}) = E(\pi(v_1), \ldots, \pi(v_{\lastproc}))$, where $\pi$ is a permutation function over individual preferences.
\end{axiom}

\begin{theorem}
\label{theorem-continuity}
No external aggregation can simultaneously respect Axioms~\ref{axiom-unanimity-cont}, \ref{axiom-anonimity-cont}, and \ref{axiom-continuity}.
\end{theorem}
\begin{proof}
Follows from~\cite{Chichilnisky1982,Baigent2011}.
\end{proof}

Note how Theorem~\ref{theorem-continuity} mirrors Theorem~\ref{theorem-proximity} in the continuous setting. At first glance, it seems that anonymity and unanimity are inherently incompatible with a notion of continuity. In fact, the \emph{generality} imposed by seemingly innocuous axioms is also fundamental to such impossibilities. Specifically, performing \emph{domain restrictions}, that is, reducing the generality of possible individual preferences, is often sufficient to escape impossibility results (please refer to~\cite{Chichilnisky1982,Baigent2011}).

The continuity requirement models our robustness to imprecision in the internal aggregation procedure. It parallels the proximity preservation property defined for ordinal preferences. In both scenarios, in light of our general axioms, we have impossibility results that seriously hinder the methodological strategy considered in this work.

\subsubsection{Inter-Protein Comparison of Energy Levels}
\label{Sec-InterProteinComparison}

Let us now increase the comparability among utility function for different proteins. The utility function for $P$ differs in origin, but not in scale, from the utility function for $Q$, taking any proteins $P \ne Q$. Formally:
\begin{axiom}
\label{axiom-originonly}
Any utility function is a strictly increasing linear transformation of another utility function: $u_P = \alpha u_Q + \beta_{PQ}$ for some \emph{fixed and unique} $\alpha > 0$ and any $\beta_{PQ}$, considering any proteins $P$ and $Q$.
\end{axiom}
Borrowing terminology from voting theory, there is no interpersonal comparisons of utility levels, but there is comparison between gains and losses among different individual preferences, since the scale is the same. For the utility preference model, we write $\samevalue{X}{Y}{I_a}{I_b}$ to denote $u_X(I_a) = u_Y(I_b)$.
Now consider the axioms below.
\begin{description}
	\item[\textit{Utility IIA.}] For any compatible $\algorithminput',\algorithminput'' \in \universeinputs$ and $\algorithmoutput',\algorithmoutput''$, if $[ \forall P_i \, \samevalue{\procinput{i}'}{\procinput{i}''}{e_a}{e_b} ]$ then we have that $\sameorder{\algorithmoutput'}{\algorithmoutput''}{e_a}{e_b}$, considering arbitrary $e_a,e_b$;
	\item[\textit{Strict Unanimity.}] For any $\algorithminput \in \universeinputs$, if $\sameorder{\procinput{i}}{\procinput{j}}{I_a}{I_b}$ for all $P_i,P_j$ then we have $\sameorder{\procinput{i}}{\algorithmoutput}{I_a}{I_b}$ for all $P_i$, considering arbitrary $I_a,I_b$; Also, if $I_a \succ_{\protinput{i}} I_b$ for an arbitrary protein $P$, then $I_a \succ_{\algorithmoutput} I_b$.
\end{description}
Interestingly, those axioms imply that the external aggregation is the \emph{utilitarian} rule~\cite{DaspremontGevers1977}:
$$ I_a \succeq_{\algorithmoutput} I_b \Leftrightarrow \sum_{P \in \cP} u_P(I_a) \ge \sum_{P \in \cP} u_P(I_b) \mathrm{.} $$

In current literature \cite{energy1,energy2,estimation1}, it is usual to perform a 1-step holistic analysis of the PDB in order to infer a certain conclusion. Social choice theory can provide interesting characterizations and consequences for different levels of comparability of preferences, or whichever concept maps suitably (e.g energy scores).

\section{Discussion and Final Remarks}
\label{Sec-DiscussionFinalRemarks}

In this work we hope to have illustrated the application and relevance of social choice theory not only to the problem and methodology that intends to infer biases in the UEF, but to other problems related to protein folding where the 2-step approach makes sense. We particularly think that the 2-step approach is sensible in light of Anfinsen's thermodynamic hypothesis, as if stable folds are the result solely of a protein's individual aminoacid sequence, considering the PDB holistically may not be always desirable.

Note that the problem considered in this work is formally \emph{weaker} than those aiming for more complete characterizations of the UEF with our 2-step approach and the same axiomatic setting. Stronger characterizations of the UEF via the method above will suffer from similar methodological limitations unless axioms are relaxed, made more specific, or otherwise changed appropriately.



Many impossibility results highlighted here rely on the \emph{generality} of the axioms. Our axioms specify, for instance, that any possible individual preference is valid. While this is attractive in an axiomatic setting, it is necessary to the Arrow's Impossibility Result and implies possibly breaking transitivity in light of May's Theorem. However, those conditions \emph{can} be avoided. In some cases (such as for Corollary~\ref{corollary-arrow} and Theorem~\ref{theorem-continuity}), a suitable domain restriction, which disallows some particular individual preferences, avoid the issues (please refer to~\cite{arrovian-imposresults,Chichilnisky1982}). If preferences are \emph{single-peaked} or \emph{quasi-transitive}, we also dodge from the impossibility results above (please refer to~\cite{Sen66,Sen69}). Given the considerations here, we invite the reader to appreciate these previous results and understand their relevance for algorithm design -- not only in the context of protein folding with the thermodynamic hypothesis, but perhaps in other similar contexts where a 2-step aggregation procedure is similarly suitable. The social choice literature is vast -- and we urge the community of computer scientists to investigate these results and their possibly immense connection and relevance to problems involving protein folding.

\bibliography{/Users/hmendes/brown/Bibliography.bib}
\bibliographystyle{plain}

\end{document}